\documentclass[conference, compsoc, letterpaper]{IEEEtran}

\usepackage{graphicx}
\usepackage{amsmath}
\usepackage{amssymb}

\DeclareGraphicsExtensions{.pdf}
\usepackage{epstopdf}
\usepackage{url}
\usepackage{xspace}
\usepackage{algorithm}
\usepackage[noend]{algpseudocode}

\newcommand{\bigO}[1]{\mbox{\mbox{O}$\paren{#1}$}\xspace}
\newcommand{\bigOmega}[1]{\mbox{\mbox{$\Omega$}$\paren{#1}$}\xspace}
\newcommand{\paren}[1]{\left(  #1 \right)}
\newcommand{\edge}[1]{(  #1 )}

\newcommand{\IN}[1]{\mbox{\mbox{IN}$[#1]$}\xspace}
\newcommand{\OUT}[1]{\mbox{\mbox{OUT}$[#1]$}\xspace}

\newcommand{\MAKESET}[1]{\mbox{\mbox{\textbf{makeset}}$(#1)$}\xspace}
\newcommand{\FIND}[1]{\mbox{\mbox{\textbf{find}}$(#1)$}\xspace}
\newcommand{\LINK}[1]{\mbox{\mbox{\textbf{link}}$(#1)$}\xspace}

\algnewcommand\algorithmicswitch{\textbf{switch}}
\algnewcommand\algorithmiccase{\textbf{case}}
\algnewcommand\algorithmicassert{\texttt{assert}}
\algnewcommand\Assert[1]{\State \algorithmicassert(#1)}%
\algdef{SE}[SWITCH]{Switch}{EndSwitch}[1]{\algorithmicswitch\ #1\ \algorithmicdo}{\algorithmicend\ \algorithmicswitch}%
\algdef{SE}[CASE]{Case}{EndCase}[1]{\algorithmiccase\ #1}{\algorithmicend\ \algorithmiccase}%
\algtext*{EndSwitch}%
\algtext*{EndCase}%

\newtheorem{thm}{Theorem}
\newtheorem{claim}{Claim}

\usepackage{tikz}
\usetikzlibrary{positioning,calc,trees}
\usepackage{forest}

\newcounter{saveenum}

\begin{document}

\title{A Simple and Efficient Algorithm for Finding Minimum Spanning Tree Replacement Edges}

\author{\IEEEauthorblockN{David A. Bader}
\IEEEauthorblockA{Department of Data Science\\New Jersey Institute of Technology\\Newark, NJ 07102\\ Email: bader@njit.edu}
\and
\IEEEauthorblockN{Paul Burkhardt}
\IEEEauthorblockA{Research Directorate\\National Security Agency\\ Fort Meade, MD 20755\\ Email: pburkha@nsa.gov}
}

\maketitle

\pagestyle{plain}

\maketitle

\begin{abstract}
Given an undirected, weighted graph, the minimum spanning tree (MST)
is a tree that connects all of the vertices of the graph with minimum
sum of edge weights.  In real world applications, network designers
often seek to quickly find a replacement edge for each edge in the
MST. For example, when a traffic accident closes a road in a
transportation network, or a line goes down in a communication
network, the replacement edge may reconnect the MST at lowest cost.
In the paper, we consider the case of finding the lowest cost
replacement edge for each edge of the MST.  A previous algorithm by
Tarjan takes \bigO{m \alpha(m, n)} time and space, where $\alpha(m, n)$ is the
inverse Ackermann's function.  Given the MST and sorted non-tree
edges, our algorithm is the first practical algorithm that runs in \bigO{m+n} time and
\bigO{m+n} space to find all replacement edges. 
% Moreover, the simplicity of the algorithm makes it easy to implement.
% assuming the MST algorithm keeps the remaining edges pre-sorted by
% weight
Additionally, since the most vital edge is the tree edge whose removal
causes the highest cost, our algorithm finds it in linear time. 
\end{abstract}

\section{Introduction}
\label{s:intro}

Let $G=(V,E)$ be an undirected, weighted graph on $n=|V|$ vertices and
$m=|E|$ edges, with weight function $w(e)$ for each edge $e \in E$. A
minimum spanning tree $T = \mbox{MST}(G)$ is a subset of $n-1$ edges
with the minimal sum of weights that connects the $n$ vertices.

In real world applications the edges of the MST often represent
roadways, transmission lines, and communication channels. When an edge
deteriorates, for example, a traffic accident shuts a road or a link
goes down, we wish to quickly find its \emph{replacement edge} to
maintain the MST. The replacement edge is the lightest weight edge
that reconnects the MST.  For example, Cattaneo et al.~\cite{CFF10}
maintain a minimum spanning tree for the graph of the Internet
Autonomous Systems using dynamic graphs. Edges may be inserted or
deleted, and a deletion of an MST edge triggers an expensive operation
to find a replacement edge of lightest weight that reconnects the MST
in \bigO{m \log n} time from the non-tree edges, or \bigO{m + n \log
  n} time when a cache is used to store partial results from previous
delete operations.

In this paper, we consider the problem of efficiently finding the
minimum cost replacement for all edges in the MST. Recomputing the MST
for each of the original tree edges is clearly too costly. The problem
is deceptively difficult. Each replacement edge must be a non-MST edge
in a fundamental cycle with the obsolete MST edge. But there are
\bigO{m} unique cycles and each cycle can have \bigO{n} MST edges so
choosing the lightest non-tree edges as replacements requires careful
planning to prevent repeatedly referencing the same MST edges. This
and related problems for updating the MST have been studied
extensively (e.g.,
\cite{SP75,CH78,Tar79,Fre85,KC96,HK97b,HDT01}).  The best algorithm is
from 1979 due to Tarjan \cite{Tar79} and runs in \bigO{m \alpha(m,n)} time and space,
where $\alpha(n,m)$ is the inverse Ackermann's function. 
But given edges sorted by weight, we show the problem can be solved in linear time and 
space using a surprisingly simple approach.

The main result of this paper is a simple and fast deterministic algorithm for the
MST replacement edge problem. Given the minimum spanning tree and
non-tree edges sorted by weight, our algorithm finds all replacement
edges in \bigO{m+n} time and \bigO{m+n} space. Although it is known that these bounds are theoretically possible using pre-sorted edges, we give the first practical algorithm. Sorted edges come free
if the MST is computed by Kruskal's algorithm. If the edge weights
have fixed maximum value or bit width, then the edges can be sorted in
linear time making our algorithm an asymptotic improvement over prior
algorithms. Our algorithm is simple and does not require queues or computing the Lowest Common Ancestor (LCA). 
It assigns a pair of numbers to each vertex denoting the order in which they are first and last visited in a depth-first search (DFS) over the MST, and for each replacement it applies linear-time path compression using the static union tree of Gabow and Tarjan \cite{GT83,GT85}.
All running times in this paper are deterministic worst-case under the Word-RAM model.

\section{Related Work}

The problem of updating the MST should a tree edge be deleted or its cost increased has been studied since the 1970s. In 1975 Spira and Pan~\cite{SP75} showed it takes \bigO{n^2} time to find the replacement for one tree edge. Then in 1978 Chin and Houck~\cite{CH78} gave a \bigO{n^2} time algorithm to find the replacements for all MST edges. This result was improved to \bigO{m\alpha(m, n)} time in 1979 by Tarjan~\cite{Tar79} and remains the best deterministic runtime for the general case of unordered, arbitrary-weight edges. The algorithm due to Tarjan uses path compression that maintains balanced tree height, which does not benefit from having edges sorted by weight. It finds the replacement edges by evaluating paths in a specially constructed directed acyclic graph in which vertices correspond to tree and nontree edges, and combining operations on vertex labels. This directed acyclic graph has \bigO{m\alpha(m, n)} vertices and therefore Tarjan's algorithm for the MST replacement edge problem takes \bigO{m\alpha(m, n)} time and space. Our method is more efficient for the case in which edges are pre-sorted by weight, taking \bigO{m+n} time and space.
 
The MST replacement edge problem can also be solved by MST sensitivity analysis, for which Tarjan~\cite{Tar82} gave a \bigO{m\alpha(m, n)} time algorithm in 1982. The MST sensitivity analysis determines the amount each edge weight can be perturbed without invalidating the MST. The sensitivity of a tree edge $e$ is the weight of an edge $f$ that is the minimum weight among the nontree edges that cross the cut induced by removing $e$, and therefore $f$ is the replacement for $e$. In 1994 Booth and Westbrook~\cite{BW94} show for planar graphs that the MST sensitivity analysis and replacement edge problems can be solved in \bigO{n} time and space. They use a depth-first search ordering of vertices similar to our method, but explicitly compute the LCA for each edge. Their method maintains two LCA-ordered lists for each leaf vertex of the tree, and therefore sorting by edge weights does not improve their runtime. 
In 1996 Kooshesh and Crawford~\cite{KC96} proposed an algorithm for the MST replacement edge problem taking \bigO{\max(C_{\mbox{mst}}, n\log n)} time, where $C_{\mbox{mst}}$ is the cost of computing the minimum spanning tree. But their runtime is not efficient and can take \bigO{n\log n} time. Their approach uses similar ideas to that of~\cite{BW94} and therefore does not improve with sorted edges in advance.

It took nearly thirty years to improve Tarjan's MST sensitivity analysis result for general graphs when in 2005 Pettie~\cite{Pet05, Pet15} improved the runtime to \bigO{m \log \alpha(m, n)} time. A more tantalizing result is Pettie showed that the MST sensitivity analysis problem is no harder than solving the MST. This implies that the MST replacement edge problem can be solved in deterministic, linear-time given an MST algorithm with the same time complexity. But a deterministic, linear-time MST algorithm is still an open problem. Interestingly, Pettie and Ramachandran~\cite{PR02b} gave an optimal MST algorithm. The 2002 Pettie-Ramachadran MST algorithm has the curious property that although it is provably optimal, the runtime is unknown but is between \bigO{m} and \bigO{m \alpha(m, n)} time.\footnote{This is due to the decision tree complexity of the MST; the height of the tree is optimal but unknown.} In the special case where edges are sorted by weight in advance, then the MST can be solved in deterministic linear time using the Fredman and Willard algorithm~\cite{FW94} by transforming real number weights to integers using relative ordering, e.g. the $i^{th}$ ordered edge gets weight $i$. Then it follows from Pettie’s reduction that the MST replacement edge problem can be solved in deterministic linear time, matching the same bounds as our algorithm. The important distinction is our algorithm is far simpler.

A related problem to MST replacement edges is that of maintaining the
MST as edges are repeatedly updated, where an update means deletion,
insertion, or weight change of an edge.  Frederickson \cite{Fre85}
gave an algorithm to maintain an MST with edge updates (deletion,
insertion, or weight change) where each update takes \bigO{\sqrt{m}}
worst-case time, and sparsification makes the bound
\bigO{\sqrt{n}}. Henzinger and King \cite{HK97b} gave an algorithm to
maintain a minimum spanning forest with edge deletions or insertions;
each update takes \bigO{\sqrt[3]{n} \log n} amortized
time. Holm~et~al. \cite{HDT01} give an algorithm for maintaining a
minimum spanning forest with edge deletions or insertions; each update
takes \bigO{\log^4 n} amortized time.
Recently, Hanauer~\emph{et al.} gave a survey of fully dynamic graph algorithms and discuss maintaining minimum spanning trees \cite{HHS22}.

\section{Algorithm}
\label{s:alg}

Given $T$ and the remaining non-tree edges $E \backslash E_T$ sorted
from lowest to highest weight, then Algorithm~\ref{alg} finds all
replacement edges for an MST in \bigO{m+n} time.  Observe that each of
the $m-n+1$ edges in $E \backslash E_T$ induces a fundamental cycle with the
edges in $T$. Then for any MST edge there is a subset of cycles
containing that edge, and the cycle induced by the lightest non-MST
edge is the replacement for it. This follows from the
\emph{Cut Property}~\cite[c.f. Theorem~23.1]{CLRS09}
where the lightest non-tree edge crossing a cut must be
in the MST if some other edge in the induced cycle is removed. Our
Algorithm~\ref{alg} finds the lightest weight cycle for each tree edge
but avoids repeatedly traversing these edges.  Since replacement edges
are found immediately after computing an MST, we can re-use the sorted
edges from Kruskal's~\cite{Kru56} MST algorithm.

\begin{algorithm}[htbp]
\caption{Linear Time MST Replacement Edges}
\label{alg}
\begin{algorithmic}[1]
\Require Graph $G$, MST edges labeled, and sorted list of non-MST edges
\Require Zero-initialized arrays $P$, IN, OUT of size $n$ indexed by all vertices $v \in V$.
\Procedure{PathLabel}{$s$, $t$, $e$} 

\If{$\IN{s}< \IN{t} <\OUT{s}$} \Comment{$s$ is ancestor of $t$} \label{alg:planStart} \label{alg:sAncestor}
\State \Return
\EndIf

\If{$\IN{t}< \IN{s} <\OUT{t}$} \label{alg:plan} \Comment{$t$ is ancestor of $s$}
\State  PLAN $\gets$ ANCESTOR, $k_1 \gets \IN{t}$, $k_2 \gets \IN{s}$

\Else

\If {$\IN{s} < \IN{t}$} 
\State  PLAN $\gets$ LEFT, $k_1 \gets \OUT{s}$, $k_2 \gets \IN{t}$ \label{alg:planLeft} \Comment{$s$ is left of $t$}
\Else
\State  PLAN $\gets$ RIGHT, $k_1 \gets \OUT{t}$, $k_2 \gets \IN{s}$ \Comment{$s$ is right of $t$}
\EndIf

\EndIf \label{alg:planEnd}

\State $v$ $\gets s$ \label{alg:startS}

\While {$k_1 < k_2$} \label{alg:travStart} \Comment{Detecting when below LCA($s$, $t$)}

\If {$\FIND{v} = v$}  \label{alg:findv} \Comment{If true, set replacement edge for $\edge{v,P[v]}$ }
%    \State $\hat{e} \gets \edge{ \mbox{$v$, P[$v$]} }$ \label{alg:inspect} \Comment{$\hat{e} \in$ MST}
% %   \If {$R_{\hat{e}} = \emptyset$}  \label{alg:noReplace} \Comment{Replacement edge has not been found yet}
%        \State $R_{\hat{e}} \gets e$  \label{alg:noReplace2} \Comment{Set the replacement edge}
% %   \EndIf

\State $R_{\edge{ v, P[v] }} \gets e$  \label{alg:replace} \Comment{Set the replacement edge}
\State $\LINK{v}$ \label{alg:link} \Comment{Union the disjoint sets of $v$ and $P[v]$}
\EndIf

\State $v \gets \FIND{v}$ \label{alg:findNext}

\Switch {PLAN} \label{alg:plan2Start}
\Case{ANCESTOR}
\State  $k_2 \gets \IN{v}$
\EndCase
\Case {LEFT}
\State  $k_1 \gets \OUT{v}$
\EndCase
\Case {RIGHT} \label{alg:caseRight1}
\State  $k_2 \gets \IN{v}$ \label{alg:caseRight2}
\EndCase
\EndSwitch \label{alg:plan2End}

\EndWhile \label{alg:travEnd}

\EndProcedure

\Statex

%\Procedure {FindReplacementEdges}{}
\State Root the MST $T$ at arbitrary vertex $v_r$ and store parents in $P$.
\State $P[v_r] \gets v_r$ \Comment{root's parent points to root}
\State \label{alg:dfs} Run DFS on $T$, setting $\IN{v}$ and $\OUT{v}$ to the counter value when $v$ is first and last visited, respectively. 

\ForAll{vertices $v \in V$} \label{alg:makeset1}
\State $\MAKESET{v}$ \label{alg:makeset2} \Comment{Initialize the disjoint sets}
\EndFor

\ForAll{edges $e \in E_T$}
\State $R_e = \emptyset$ \Comment{Initialize the replacement edges}
\EndFor

\For{$k \gets 1 \, .. \, m-n+1$}  \label{alg:scanStart}  \Comment{Scan the $m-n+1$ sorted non-MST edges}

\State $\edge{ v_i, v_j} \gets e_k$

\State \Call{PathLabel}{$v_i, v_j, \edge{ v_i, v_j }$} \label{alg:pathlabel1}
\State \Call{PathLabel}{$v_j, v_i, \edge{ v_i, v_j }$} \label{alg:pathlabel2}
\EndFor \label{alg:scanEnd}

\end{algorithmic}
\end{algorithm}

The major steps of our approach are 1) assign the parent and the first and last visited numbers to each vertex according to depth-first search over $T$ 2) traverse the fundamental cycle induced by each non-tree edge in order of ascending weight 3) compress paths to skip edges already assigned replacements.
With \bigO{m} non-tree edges and \bigO{n} edges in each cycle, the na\"{\i}ve approach has
\bigOmega{mn} time complexity. This paper introduces an algorithm that
reduces the cost to \bigO{m+n} time by a novel use of a special case
of the disjoint set union data structure.  We use the disjoint sets
for fast path compression based on the Gabow-Tarjan static union tree
method \cite{GT83,GT85}.

Algorithm~\ref{alg} first roots the MST at an arbitrary vertex $v_r$
and initializes a parent array $P$.  Next, each vertex $v \in V$ is
visited during a depth-first search traversal from the root, and
the value of $P[v]$ is set to its respective parent vertex from the
traversal order.  For the root $v_r$, its parent $P[v_r]$ is set to
$v_r$.  Our approach uses another innovation that alleviates the need
to find the lowest common ancestor vertex in the rooted MST for each
non-tree edge. To do so, we use a pair of vertex-based values, \IN{v}
and \OUT{v}, which are assigned as follows.  During the depth-first
traversal of the rooted tree, a counter is incremented for each step
in the traversal (up or down edges). When the traversal visits $v$ the
first time during a traversal down an edge, \IN{v} is assigned the
current counter value.  When the traversal backtracks up an edge from
vertex $v$, \OUT{v} is then assigned the current counter value.
% Then using an Euler-tour tree, we record the first and last time
% each vertex is visited in a depth-first search from the root.  The
% Euler-tour tree is a doubly-linked tree $T^{\prime}$ constructed by
% converting each edge in $T$ to a pair of directed edges in
% $T^{\prime}$. For example, the edge from a vertex $v$ to its parent
% $P[v]$ induces a pair of directed edges $\edge{v, P[v]}$ and
% $\edge{[P[v],v]}$ in the Euler-tour tree $T^{\prime}$.  A
% depth-first traversal of $T^{\prime}$ labels each directed edge with
% a unique value of the traversal visit order. In this way, the value
% $I$ of edge $\edge{P[v], v}$ (which is called \IN{v}) represents the
% first time vertex $v$ is reached in the traversal down, and the
% value of edge $\edge{v, P[v]}$ (which is called \OUT{v}) represents
% the traversal completing with vertex $v$ and its subtree.
% Figure~\ref{fig:dfs} shows the depth-first traversal of the
% Euler-tour tree $T^{\prime}$. Constructing $T^{\prime}$ and its
% traversal takes \bigO{n} time.
With a minor modification to our algorithm we can also employ the conventional pre- and post-order numbers from depth-first search to detect the LCA. 
Our ordering is chosen for convenience because it requires only a single increment on the visit order, which simplifies tracking the traversal as an increasing sequence of visit numbers.
%It is also helpful in keeping track of the traversal since it induces a total ordering over the pre- and post-order numbers. 
For example, using pre($v$), post($v$) to denote respectively the pre- and post-order of a vertex $v$, then when identifying if $s$ is an ancestor of $t$ we can replace the conjunction pre($s$) $<$ pre($t$) and post($t$) $<$ post($s$) with the straightforward inequality sequence IN[$s$] $<$ IN[$t$] $<$ OUT[$t$] $<$ OUT[$s$].
%Figure~\ref{fig:dfs} shows the depth-first traversal of the MST. 
%The traversal of the minimum spanning tree $T$ takes \bigO{n} time.

%\begin{figure}[htbp]
%\centering
%%\includegraphics[width=0.4\textwidth]{fig1}
%\includegraphics[width=0.3\textwidth]{fig1.pdf}
%\caption{Depth-first traversal of the minimum spanning tree $T$, setting each parent $P[v]$, and \IN{v}, and \OUT{v} to the counter when vertex $v$ is first and last visited, respectively. The black edges are parent pointers representing the tree $T$.}\label{fig:dfs}
%\end{figure}

The $m-n+1$ remaining edges in $E \backslash E_T$ are scanned in
ascending order by weight, inspecting the tree edges in each
corresponding fundamental cycle. In this order, the first time a tree
edge $e$ is included in a fundamental cycle, its replacement $R_{e}$
is set to the non-tree edge from that cycle. As we will describe, the
disjoint sets provide subpath compression as replacement edges
are assigned to MST edges. In Algorithm~\ref{alg}, the disjoint
sets are updated through the \textbf{makeset}, \textbf{find}, and
\textbf{link} functions.

\begin{figure}[htbp]
\centering
\begin{forest}
  for tree={
    calign=fixed edge angles,
    draw,
    circle,
    minimum size=16,
    inner sep=0
    },
  [$v_r$
    [,name=head
      [$\ldots$,draw=none,edge=dashed
        [$z$,edge=dashed
          [
            [$\ldots$,draw=none,edge=dashed
              [$s$,name=tail,edge=dashed
                []
                [$\ldots$,draw=none,edge={draw=none}]
                []
              ]
              %[,phantom]
              [,phantom]
            ]
            %[,phantom]
            [,phantom]
          ]
          [$\ldots$,draw=none,edge=dashed
            [,phantom]
            [$t$,edge=dashed,%l*=2
              []
              [$\ldots$,draw=none,edge={draw=none}]
              []
            ]
          ]
        ]
        %[,phantom]
        [,phantom]
      ]
      %[,phantom]
      [,phantom]
    ]
    [,phantom]
    [$\ldots$,draw=none,edge=dashed]
  ]
  \coordinate (h) at ([xshift=-1cm] head);
  \coordinate (t) at ([xshift=-1cm] tail);
  \draw[->,very thick] (t) to node[left] {$w$} (h);
\end{forest}

\caption{The PathLabel algorithm detects when vertex $w$ on the path
  from $s$ to the root $v_r$ is an ancestor of the vertex
  $z=\mbox{LCA}[s,t]$, without determining $z$.}\label{fig:PathLabel}
\end{figure}
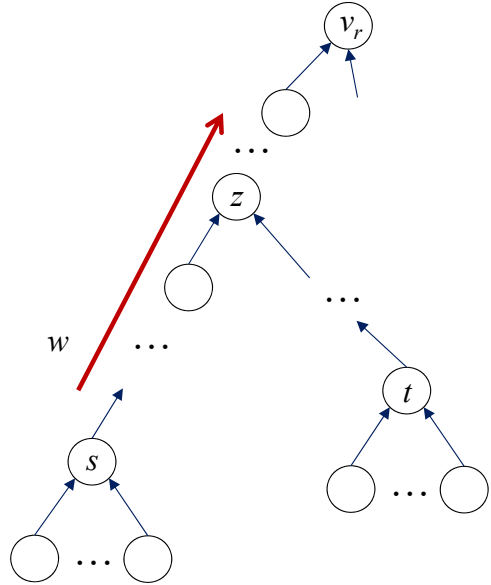

For each non-tree edge $\edge{s, t}$, if vertex $t$ is a descendant
of $s$, (if and only if $\IN{s} < \IN{t} < \OUT{t} < \OUT{s}$), we
make a single \textsc{PathLabel} call for the edges from $t$ up to
$s$.  Since $\IN{t} < \OUT{t}$, we simplify this check in
Algorithm~\ref{alg}, line~\ref{alg:sAncestor}, to $\IN{s} < \IN{t} <
\OUT{s}$.

Otherwise, two calls are made to \textsc{PathLabel}, corresponding to
inspecting the \emph{left} and \emph{right} paths of the cycle from
$s$ and $t$, respectively, that would meet at the LCA of $s$ and $t$
in the tree.  We assume, without loss of generality, that $s$ is
visited in the depth-first search traversal before $t$.  Let's call $z
= \mbox{LCA}[s, t]$.  It is useful to use $z$ in describing the
approach, yet we never actually need to find the LCA $z$.  We know
$\IN{z} < \IN{s} < \OUT{s} < \IN{t} < \OUT{t} < \OUT{z}$ by definition
of the depth-first traversal.  As illustrated in
Figure~\ref{fig:PathLabel}, consider a vertex $w$ that lies on the
path from $s$ to the root $v_r$. Vertex $w$ must either lie on the
path from $s$ to the LCA $z$ (where $\OUT{w} < \IN{t}$), or from $z$
to the root $v_r$ (where $\OUT{w} > \IN{t}$). We use this fact to
detect when \textsc{PathLabel} reaches the LCA without computing it.

As mentioned earlier, the disjoint sets provide subpath compression as
replacement edges are assigned to MST edges. Initially, each vertex is
placed in its own set.  While traversing edges in a cycle that have
not yet been assigned a replacement, the disjoint sets compress
the subpath by uniting the sets corresponding to each vertex and its
parent in the tree, thereby ensuring that MST edges are traversed at
most once.

The conventional union heuristic with path compression for the
disjoint set union problem would not lead us to a linear-time
algorithm.  Gabow and Tarjan \cite{GT83,GT85} designed a linear-time
algorithm for the special case where the structure of the unions,
called the \emph{Union Tree}, is known in advance. The Gabow-Tarjan
approach executes a sequence of $m$ union and find operations on $n$
elements in \bigO{m+n} time and \bigO{n} space. The functions are
$\MAKESET{v}$ that initializes $v$ into a singleton set with label
$v$, $\FIND{v}$ that returns the label of the set containing $v$, and
$\LINK{v}$ that unites the sets $v$ and $P[v]$, where $P[v]$ is the
parent of $v$ in the union tree, and gives it the label of the set
containing $P[v]$.

For the path compression used in our MST replacement edge algorithm,
the structure of unions is known in advance; that is, the union tree
is equivalent to the MST. Hence, we use the Gabow-Tarjan approach for
the disjoint sets and path compression.

There are cases when the algorithm may terminate prior to scanning the
entire list of edges. This observation leads to a faster
implementation that still runs in linear time.  A \emph{bridge} edge
of a connected graph is defined as an edge whose removal disconnects
the graph. Clearly, bridge edges will always be included in the MST
and will not have a replacement edge in the solution. Tarjan
\cite{Tar74} shows that counting the number of bridges in the graph
$G$ takes \bigO{m+n} time.  Thus, Algorithm~\ref{alg} may terminate
the scanning of remaining edges once $n-1-k$ replacement edges are
identified, where $k$ is the number of bridges in $G$.

% \clearpage

\subsection{Example}

In this section we give a simple walk-through of the algorithm on the
graph in Figure~\ref{exampleGraph}. This example exercises all three
plans in the algorithm.

\begin{figure}[h]
\centering
\begin{tikzpicture}
[vertex/.style={circle,draw=black,node distance=1.5,minimum size=16,inner sep=0}]
\node [vertex] (a) [label=above:{\tiny $(13,14)$}] {a};
\node [vertex] (b) [label=above:{\tiny $(9,10)$}] at ([shift=({25:2})]a) {b};
\node [vertex,style={thick},line width=2.5] (c) [label=left:{\tiny $(1,16)$}] at ([shift=({-85:1.5})]b) {c};
\node [vertex] (d) [label=below:{\tiny $(2,15)$}] at ([shift=({-65:2})]a) {d};
\node [vertex] (e) [label=above:{\tiny $(4,11)$}] at ([shift=({-30:2})]b) {e};
\node [vertex] (f) [label=below:{\tiny $(3,12)$}] at ([shift=({0:2.5})]d) {f};
\node [vertex] (g) [label=above:{\tiny $(6,7)$}] at ([shift=({5:2.85})]b) {g};
\node [vertex] (h) [label=above:{\tiny $(5,8)$}] at ([shift=({-40:2})]g) {h};
% MST edges in order
\draw [style={thick}, line width=2.5] (d) -- node [below] {1} (f);
\draw [style={thick}, line width=2.5] (e) -- node [above] {2} (h);
\draw [style={thick}, line width=2.5] (g) -- node [above] {3} (h);
\draw [style={thick}, line width=2.5] (e) -- node [left] {4} (f);
\draw [style={thick}, line width=2.5] (b) -- node [below] {5} (e);
\draw [style={thick}, line width=2.5] (c) -- node [left] {7} (d);
\draw [style={thick}, line width=2.5] (a) -- node [left] {9} (d);
% non-MST edges
\draw (a) -- node [above] {10} (b);
\draw (e) -- node [right] {6} (g);
\draw (b) -- node [above] {8} (g);
\draw (c) -- node [below] {11} (e);
\draw (b) -- node [left] {12} (c);
\draw (f) -- node [right,xshift=1ex] {13} (h);
\end{tikzpicture}

\caption{An example graph on 8 vertices $(a, \ldots, h)$ and 13 weighted edges. The MST root vertex $c$ and MST edges are highlighted by thicker lines.}
\label{exampleGraph}
\end{figure}
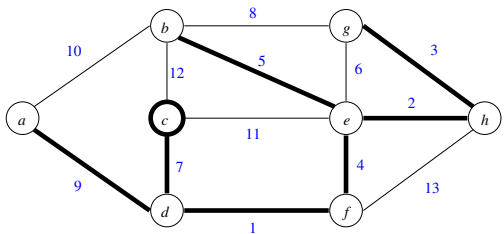

The MST edges have weights 1, 2, 3, 4, 5, 7, and 9, and say the
root of the MST tree is vertex $c$. 
The vertices in Figure~\ref{exampleGraph} have been labelled with (IN,OUT) numbers assigned by DFS over the MST edges with branching order by ascending edge weight.
In the following walk-through of
the algorithm, the reader should note that all vertices retain their
original parents. Also we remark that all walks are in order from
descendent to ancestor or from \emph{last} to \emph{first} in DFS
order.
 
Then in sorted, non-MST edge order we begin with the $\edge{g,e}$ edge
at line~\ref{alg:pathlabel1}.
\begin{enumerate}
\item Vertex $e$ is the ancestor of $g$, then at line~\ref{alg:plan}
  we get the ancestor (ANCESTOR) plan with $k_1=\IN{e}$, $k_2=\IN{g}$ and
  thus $k_1 < k_2$.
\item The cycle traversal begins with $g$
  (line~\ref{alg:startS}). Since $g$ has not yet been visited then
  line~\ref{alg:replace} assigns the current non-MST edge $\edge{g,e}$
  to $\edge{g,h}$, where $h$ is the parent of $g$.
\item The disjoint sets are linked (line~\ref{alg:link}) so $g$'s
  disjoint set gets $h$'s label. This compresses the subpath $(g,h)$.
\item The next vertex is $h$ since it is the parent of $g$
  (line~\ref{alg:findNext}) and then $k_2$ is updated to \IN{h}.
\item Continuing the traversal with $h$ (line~\ref{alg:travStart}),
  again line~\ref{alg:replace} assigns $\edge{g,e}$ to $\edge{h,e}$
  where $e$ is the parent of $h$.
\item The disjoint sets are linked (line~\ref{alg:link}) so $h$'s
  disjoint set gets $e$'s label. This compresses the subpath
  $(g,h,e)$.
\item Now the next vertex is $e$ so $k_2$ gets \IN{e} making it equal
  to $k_1$, thus ending the while loop.
\item The oppositely-oriented edge $\edge{e,g}$ input at
  line~\ref{alg:pathlabel2} is not processed because $e$ is the
  ancestor of $g$ and we have already followed the path from
  descendent to ancestor.  \setcounter{saveenum}{\value{enumi}}
\end{enumerate}
 
The next non-MST edge is $\edge{b,g}$ and input at
line~\ref{alg:pathlabel1}, since it happens that $g$ was reached before $b$ in the DFS.
\begin{enumerate}
  \setcounter{enumi}{\value{saveenum}}
\item We get the RIGHT branch plan with $k_1=\OUT{g}$, $k_2=\IN{b}$
  and so again $k_1 < k_2$.
\item The traversal begins with $b$ (line~\ref{alg:startS}) and since
  $b$ has not yet been visited then $\edge{b,e}$ gets the non-MST edge
  $\edge{b,g}$, where $e$ is the parent of $b$.
\item The disjoint set is linked (line~\ref{alg:link}) so $b$'s
  disjoint set gets $e$'s label and the subpath $(b,e)$ is compressed.
\item The next vertex is $e$ (the parent of $b$) and thus $k_2$ is
  updated to \IN{e} (lines~\ref{alg:caseRight1}-\ref{alg:caseRight2})
  making $k_2 < k_1$ and thus ending the while loop.
\item The oppositely-oriented edge $\edge{g,b}$ is input at
  line~\ref{alg:pathlabel2}.
\item We get the LEFT branch plan (line~\ref{alg:planLeft}) with
  $k_1=\OUT{g}$, $k_2=\IN{b}$ so $k_1 < k_2$ and start the traversal
  with $g$.
\item Now observe that the disjoint sets had previously compressed the
  subpath $(g,h,e)$. Thus $\FIND{g} \ne g$. This jumps the walk to the
  LCA, which is vertex $e$, and updates $k_1$ to \OUT{e} to end the
  while loop.
\end{enumerate}

Observe for edge $\edge{b,g}$ that if $b$ were reached before $g$ in
the DFS, it would have finished earlier but all subpaths would have
been compressed as before. We leave it as an exercise for the reader
to finish the algorithm on the remaining non-MST edges.

\subsection{Proof of correctness}

\begin{claim}
The lowest weight non-MST edge that induces a cycle containing an MST
edge $e$ is the replacement for $e$. This follows from the 
\emph{Cut Property}~\cite[c.f. Theorem~23.1]{CLRS09}.
\label{claim1}
\end{claim}

\begin{claim}
Algorithm~\ref{alg} traverses the cycle induced by a non-MST edge from
descendent to ancestor, and stops at the LCA (in the case that the LCA
is different from $s$ and $t$).
\label{claim2}
\end{claim}
\begin{proof}
Observe that the parent is set for each vertex in DFS order so that
the traversal carried out by lines
\ref{alg:travStart}--\ref{alg:travEnd} follows a single path from
descendent to ancestor. The path is an upwards traversal of the
compressed subpaths in the disjoint sets. For each $\edge{s, t}$ edge,
$s$ may be the ancestor of $t$ or vice versa, or the LCA is neither
$s$ nor $t$. The lines \ref{alg:planStart}--\ref{alg:planEnd} always
set the starting vertex in the traversal of the cycle so that it
proceeds from descendent to ancestor as follows.

If $s$ is the ancestor of $t$ then no traversal is made because
line~\ref{alg:sAncestor} returns.  If $t$ is the ancestor of $s$, then
the traversal begins with $s$ at line~\ref{alg:startS} and each
traversal up using the disjoint sets leads to $t$.  Otherwise, there
is an LCA and from lines \ref{alg:pathlabel1}--\ref{alg:pathlabel2}
each branch is traversed from $s$ and $t$ up to the LCA.  The subpath
compression using disjoint sets occurs at line~\ref{alg:link}.  The
linking unites all sets corresponding to vertices in the tree
traversal from $s$ and $t$ up to the LCA.
\end{proof}

\begin{claim}
Algorithm~\ref{alg} traverses only those edges in the unique cycle
induced by a given non-MST edge.
\label{claim3}
\end{claim}
\begin{proof}
We prove this using a loop invariant for a single cycle. Let $\edge{s, t}$
be a non-MST edge and denote the cycle it induces by $s,v_i, v_{i+1},
\ldots ,t,s$.

The loop invariant is: $v$ at the start of the while loop at
lines \ref{alg:travStart}--\ref{alg:travEnd} must be a vertex in the
cycle induced by $\edge{s, t}$.

The base step holds trivially since the starting vertex is $s$.

The inductive step maintains the loop invariant as follows. At each
iteration the disjoint sets of each vertex and its parent are united
and by Claim~\ref{claim2} this vertex must be a predecessor in the
path from descendent to ancestor. Thus every iteration produces the
sequence $v_i, v_{i+1}, \ldots, v_p$ where $v_p$ is either $t$ or an
LCA of $s$ and $t$.  By Claim~\ref{claim2}, the traversal cannot go
above the LCA of $s$ and $t$.

Termination of the loop is determined by new values for either $k_1$
or $k_2$ between lines \ref{alg:plan2Start}--\ref{alg:plan2End}. If
the case was that $t$ was the ancestor of $s$, then $k_2$ decreases in
value as the path traversal using disjoint sets approaches
$t$. Otherwise the LCA is neither $s$ nor $t$ and if $s$ is visited
before $t$ in DFS order, then it is in the \emph{left} branch and
$k_1$ increases in value as the upwards path traversal using disjoint
sets approaches $t$, otherwise we have the \emph{right} branch and
similarly the loop ends as the path traversal using disjoint sets
moves towards the other endpoint.
\end{proof}

\begin{thm}
Given the Minimum Spanning Tree for an undirected, weighted graph $G =
(V, E)$, and non-tree edges sorted by weight, then Algorithm~\ref{alg}
correctly finds all minimum cost replacement edges in the Minimum
Spanning Tree of $G$.
\end{thm}
\begin{proof}
First observe that all non-MST edges are processed in ascending order
by weight between lines \ref{alg:scanStart}--\ref{alg:scanEnd}. Then
the $\edge{s, t}$ edge that induces the first cycle to contain an MST
edge must be the replacement edge for that MST edge following
Claim~\ref{claim1} and the order of processing. This is carried out by
line~\ref{alg:replace}, hence each MST edge gets the first non-MST
edge that induces a cycle containing it.

It follows from Claim~\ref{claim3} and the loop over all non-MST edges
at lines~\ref{alg:scanStart}--\ref{alg:scanEnd} that all MST edges in
a cycle will get a replacement edge.

At the end of a cycle, the traversed edges in the subpath are
compressed with each parent set by linking the disjoint sets so that
any edge from this cycle cannot be traversed again because it has been
assigned a replacement edge.
\end{proof}

\subsection{Complexity analysis}

\begin{claim}
Algorithm~\ref{alg} updates disjoint sets in \bigO{m+n} time and \bigO{n} space. 
\label{claim4}
\end{claim}
\begin{proof}
The Gabow-Tarjan disjoint sets use \bigO{n} \textbf{makeset}
operations (lines~\ref{alg:makeset1}-\ref{alg:makeset2}), one for each vertex $v \in V$; and \bigO{n} \textbf{link}
operations (line~\ref{alg:link}) since there are at most $n-1$ replacement edges.  For each
non-tree edge, there are at most two \textbf{find} operations at the
start and end of each of the two \textsc{PathLabel} calls,
corresponding with the initial \FIND{s} (line~\ref{alg:findv}) and the
final \FIND{v} (line~\ref{alg:findNext}) that returns a label of either
$t$ or an ancestor of $t$. Hence these contribute to at most $4
(m-n+1) = \bigO{m}$ \textbf{find} operations.  Every other
\textbf{find} precedes a \textbf{link} operation, so there are \bigO{n}
of these \textbf{find} operations. Therefore, Algorithm~\ref{alg} uses
\bigO{m} \textbf{find} operations. 

The union tree is equivalent to
the MST tree. Hence, Algorithm~\ref{alg} uses the special case of
disjoint set union when the union tree is known in advance.  Using the
Gabow-Tarjan disjoint set union, thus, takes \bigO{m+n} time and
\bigO{n} space.
\end{proof}

\begin{thm}
Given the Minimum Spanning Tree for an undirected, weighted graph $G =
(V, E)$, and non-tree edges sorted by weight, then Algorithm~\ref{alg}
finds all minimum cost replacement edges of the Minimum Spanning Tree
of $G$ in \bigO{m+n} time and \bigO{m+n} space.
\end{thm}
\begin{proof}
Let $T$ be the Minimum Spanning Tree of $G$. Initializing all values
in the parent array $P$ takes \bigO{n} time.  Since there are $n-1$
edges in $T$ then running DFS on $T$
(line~\ref{alg:dfs}) to initialize the IN and OUT arrays takes
\bigO{n} time. Initializing the replacement edges of the MST edges
takes \bigO{n} time.

There are $m - n + 1 = \bigO{m}$ non-MST edges read in ascending order
by weight, taking \bigO{m} time.  For each non-MST edge, it was
established by Claim~\ref{claim3} that the algorithm can only
reference edges in the fundamental cycle induced by that non-MST
edge. These edges are traversed only once as follows.

The algorithm walks each fundamental cycle in the same direction from
descendant to ancestor, as imposed by the DFS ordering set in the IN
and OUT arrays.  On visiting a vertex $v$, if $v$'s set label equals
$v$ then the edge $\edge{v,P[v]}$ has not been visited before,
otherwise it violates the path compression at lines
\ref{alg:findv}-\ref{alg:link}.  When an edge $\edge{v,P[v]}$ gets a
replacement (line~\ref{alg:replace}), the disjoint set corresponding
with parent $P[v]$ is united with $v$'s set (line~\ref{alg:link})
using the Gabow-Tarjan disjoint set \textbf{link} operation.  The
label of the new set is the root of the induced subtree of the MST.
Therefore when a vertex $v$ is first visited, \LINK{v} results in the
set label being the label of the set containing $P[v]$. On completing
the walk along the cycle, the set label will be the label of the set
containing the LCA. Then subsequent \FIND{v} operations return the
most recent root of the subtree containing $v$.  Thus this sequence of
disjoint set unions perform path compression on tree edges with
assigned replacements. The compressed path decreases the traversal
length of subsequent walks beginning at vertices lower in the DFS
ordering by skipping over tree edges already with assigned replacement
edges.

It follows from Claim~\ref{claim3} and this specific ordering of the
disjoint set labels that the algorithm cannot follow a path that does
not close the cycle.  Then because of path compression only \bigO{m}
edges are traversed, taking \bigO{m} time. Claim~\ref{claim4}
establishes \bigO{m+n} time and \bigO{n} space for all disjoint set
operations.  Hence it takes \bigO{m+n} time to find all replacement
edges in $T$. The data structures are simple arrays and Gabow-Tarjan
disjoint set union data structures, taking \bigO{n} space, and all
non-MST edges take \bigO{m} space. Therefore it takes \bigO{m+n} time
and \bigO{m+n} space as claimed.
\end{proof}

\section{Most Vital Edge}

The most vital edge of a connected, weighted graph $G$ is the edge
whose removal causes the largest increase in the weight of the minimum
spanning tree \cite{HJL91}.  When the graph contains bridges (which
can be found in linear time \cite{Tar74}), the most vital edge is
undefined.  The input for methods to find the most vital edge include both the graph and the edges sorted by weight. 
Hsu et al.~\cite{HJL91} designed algorithms to find the
most vital edge in \bigO{m \log m} and \bigO{n^2} time.  Iwano and
Katoh \cite{IK93} improve this with \bigO{m+n \log n} and \bigO{m
  \alpha(m,n)} time algorithms.  Suraweera et al.~\cite{SMB95} prove
that the most vital edge is in the minimum spanning tree.  Hence, once
Algorithm~\ref{alg} finds all replacement edges of the minimum
spanning tree, the most vital edge takes \bigO{n} time by simply
finding the tree edge with maximum difference in weight from its
replacement edge.  Thus, our approach will also find the most vital
edge in \bigO{m+n} time, and is the first linear algorithm for finding
the most vital edge of the minimum spanning tree given the non-tree
edges sorted by weight.

\bibliographystyle{plain}
\bibliography{parallel}
\end{document}